\documentclass[12pt]{article}
\usepackage{amsmath}
\usepackage{amsthm}
\usepackage[normalem]{ulem}

\usepackage{amssymb}
\usepackage{amsmath}
\usepackage{bbm}
\usepackage{multirow}
\usepackage{latexsym}
\usepackage{epsfig}
\usepackage{graphicx}
\usepackage{epstopdf}
\usepackage[usenames]{color}
\usepackage{soul}

\def\boxit#1{\vbox{\hrule\hbox{\vrule\kern6pt
          \vbox{\kern6pt#1\kern6pt}\kern6pt\vrule}\hrule}}

\usepackage{natbib}

\newcommand{\bfb}{{\bf b}}

\newcommand{\bfp}{{\bf p}}

\newcommand{\bfy}{{\bf y}}

\newcommand{\bfB}{{\bf B}}

\newcommand{\bfD}{{\bf D}}

\newcommand{\bfG}{{\bf G}}
\newcommand{\bfH}{{\bf H}}
\newcommand{\bfI}{{\bf I}}

\newcommand{\bfL}{{\bf L}}
\newcommand{\bfM}{{\bf M}}

\newcommand{\bfP}{{\bf P}}
\newcommand{\bfQ}{{\bf Q}}
\newcommand{\bfR}{{\bf R}}

\newcommand{\bfV}{{\bf V}}
\newcommand{\bfW}{{\bf W}}
\newcommand{\bfX}{{\bf X}}
\newcommand{\bfY}{{\bf Y}}
\newcommand{\bfZ}{{\bf Z}}

\newcommand{\ysp}{{\widetilde{\bfy}}}
\newcommand{\xsp}{{\widetilde{\bfX}}}

\newcommand{\bfbeta}{\mbox{\boldmath $\beta$}}
\newcommand{\bfepsilon}{\mbox{\boldmath $\epsilon$}}

\newcommand{\bfeta}{\mbox{\boldmath $\eta$}}
\newcommand{\bftheta}{\mbox{\boldmath $\theta$}}

\newcommand{\bfphi}{\mbox{\boldmath $\phi$}}
\newcommand{\bfxi}{\mbox{\boldmath $\xi$}}

\newcommand{\bfSigma}{\mbox{\boldmath $\Sigma$}}

\newcommand{\bfOmega}{\mbox{\boldmath $\Omega$}}

\newcommand{\bfzeta}{\mbox{\boldmath $\zeta$}}

\newcommand{\diag}{\mathrm{diag}}

\newtheorem{theorem}{Theorem}[section]

\newcommand{\bfone}{{\bf 1}}
\newcommand{\bfzero}{{\bf 0}}

\begin{document}
\thispagestyle{empty} \baselineskip=28pt

\begin{center}
{\LARGE{\bf A novel reference prior for Gaussian hierarchical models with intrinsic conditional autoregressive random effects}}
\end{center}

\baselineskip=12pt

\vskip 2mm
\begin{center}
Marco A. R. Ferreira\footnote{\baselineskip=10pt Department of Statistics, Virginia Tech, Blacksburg, VA, USA, marf@vt.edu}
\end{center}
%
%
%
%

\begin{center}
{\large{\bf Abstract}}
\end{center}
\baselineskip=12pt

We develop a novel reference prior for Gaussian hierarchical models with intrinsic conditional autoregressive (ICAR) random effects. 
This is particularly important in the context of objective Bayes variable selection with sample size $n$ and $k$ regressors.
In this context, a previously published reference prior requires the computation of spectral decompositions of two $n$-dimensional matrices for each model under consideration. 
As a consequence, for variable selection the computational cost of this previous reference prior grows as $O(n^3 2^k)$. 
In contrast, our novel reference prior requires the computation of the spectral decomposition of one $n$-dimensional matrix that can be used for all models under consideration. 
Thus, the computational cost of our novel reference prior grows much slower as $O(n^3)$. 
Hence, computational savings can be substantial, e.g. in a problem with 10 regressors, when compared to the previously published reference prior,  computations based on our novel reference prior are more than 1000 times faster.
We provide a proof of the equivalence of the two priors. 
A simulation study shows that, while both reference priors provide equivalent variable selection results, for large sample sizes computations based on our novel prior are several orders of magnitude faster. 
Finally, the utility of our novel reference prior is illustrated with a spatial regression study of county-level median household income on socio-economic regressors for 3108 counties in the contiguous United States.
%
%
%

\baselineskip=12pt
\par\vfill\noindent

\par\medskip\noindent

\clearpage\pagebreak\newpage \pagenumbering{arabic}
\baselineskip=24pt

\section{Introduction}

Hierarchical models with conditional autoregressive (CAR)  \citep{besag1974} and intrinsic conditional autoregressive (ICAR) \citep{besag1991} random effects have been applied in a wide variety of fields such as ecology \citep{VerHoef2018}, data on network graphs \citep{kolaczyk2020statistical}, complex survey data \citep{mercer2015}, disease mapping \citep{reich2006, Jin2007, Lee2011, Goicoa2018}, and neuroscience \citep{liu2016, lyu2024spatial}. For situations when little or no prior information about the parameters is available, \citet{keef:ferr:fran:2019} developed a reference prior for Gaussian hierarchical models with ICAR random effects. Henceforth, we refer to this as the KFF prior. 
When compared to widely used gamma priors \citep{best:etal:1999,lee2013}, Bayesian analysis based on the KFF prior performs favorably in terms of mean squared estimation error and credible intervals with frequentist coverage close to nominal  \citep{keef:ferr:fran:2019,Ferreira2021}. Also based on the KFF prior, \citet{Porter2024} developed an objective Bayes model selection approach for Gaussian hierarchical models with ICAR random effects. 
Compared to three widely used criteria \citep{wata:2010,spiegelhalter2002, cele:etal:2006}, the model selection approach of \citet{Porter2024} provides much better performance in terms of identifying the correct spatial structure and regressors. 
Despite its excellent statistical properties, the KFF prior suffers from a serious drawback: for model selection problems with sample size $n$ and $k$ regressors, the KFF prior computational cost increases fast as $O(n^3 2^k)$.

To address this drawback, here we propose a new reference prior that is much faster to compute than the KFF prior. While the KFF prior was based on a theorem provided by \citet{deoliveira2007}, our novel prior is based on a theorem provided by \citet{berg:oliv:sans:2001}. 
The KFF prior requires the computation of spectral decompositions of two $n$-dimensional matrices for each model under consideration, leading to the $O(n^3 2^k)$ growth in computational cost. 
In contrast, our novel reference prior requires the computation of the spectral decomposition of one $n$-dimensional matrix that can be used for all models under consideration. 
Thus, the computational cost of our novel reference prior grows much slower as $O(n^3)$. 
Finally, we provide a theorem that states the mathematical equivalence of the KFF prior and our novel prior. Therefore, our novel prior inherits all the excellent statistical properties of the KFF prior, and leads to much faster computations.

In addition to the novel reference prior, we propose to accelerate computations using spectral computations. Specifically, 
\citet{Ferreira2021} has proposed to accelerate computations for Gaussian hierarchical models with ICAR spatial random effects by transforming the hierarchical model from the spatial domain
to the spectral domain. 
Here, we extend the proposal by \citet{Ferreira2021} to accelerate computations for objective Bayes model selection. 
In this context, in contrast to the usual $O(n^3)$ computations of determinants and matrix inverses, computations in the spectral domain cost $O(n)$ operations.

To compare the performance of the KFF prior and our novel prior in objective Bayes model selection computations, we perform a simulation study to investigate how the computational costs of the two priors increase with the sample size. 
This simulation study shows that, while both reference priors provide equivalent variable selection results, the computational cost of the KFF prior increases much faster  with sample size $n$. As a consequence, 
for large sample sizes computations based on our novel prior are several orders of magnitude faster. 

Finally, we illustrate the use of our novel reference prior with an analysis of median household income per county in the contiguous United States. 
Specifically, we consider 11 candidate socio-economic regressors and compute the posterior model probabilities of the $2^{11}=2048$ possible models. 
Based on these posterior model probabilities, we report the posterior inclusion probabilities of each regressor.  
With 3108 counties and 11 candidate regressors, computations based on the KFF prior in a standard laptop as implemented in the R package ref.ICAR version 2.0.2 \citep{port:etal:2025} would take several months. 
In contrast, computations based on our novel reference prior  take 27.3 minutes. 

The remainder of the article is organized as follows. 
Section~\ref{sec:model} introduces notation and reviews hierarchical models with ICAR random effects, the KFF  prior, and the objective Bayes model selection approach proposed by \citet{Porter2024}.
Section~\ref{sec:spectral-domain-computations} presents fast and scalable spectral domain computations for objective Bayes model selection. 
Section~\ref{sec:new-reference-prior} presents the novel reference prior and a theorem that establishes its equivalence to the KFF prior. 
Section~\ref{sec:computational-time} presents a simulation study to compare the computational cost of the KFF prior and the novel reference prior. 
Section~\ref{sec:application} illustrates the utility of our novel reference prior for model selection for large spatial datasets with a spatial regression of the logarithm of median household income on socio-economic variables per county in the contiguous United States. 
Finally, Section~\ref{sec:discussion} provides a brief discussion about our main findings and possible future research directions. 
For convenience of exposition, all proofs appear in the Appendix.


\section{Hierarchical models with ICAR random effects}\label{sec:model}

\subsection{The hierarchical model}
\label{sec:modelsub}

We focus on spatial areal data observed across a geographical region partitioned into $n$ disjoint subregions, indexed $1, \ldots, n$. Adopting the notation of \citet{keef:ferr:fran:2018}, we assume a neighborhood structure where $N_j$ denotes the set of neighbors of subregion $j$, for $j=1,\ldots,n$.
Specifically, we consider the model 
\begin{equation}
\label{eqn:modelequation}
\bfy = \bfX \bfbeta + \bfphi + \bfepsilon, 
\end{equation}
where $\bfy$ is an $n$-dimensional vector of observed responses, $\bfX$ is an $n \times p$ matrix of regressors, and $\bfbeta$ is a $p$-dimensional vector of regression coefficients. The matrix $\bfX$ has a column of ones denoted by $\bfone_n$, and a column for each of $k=p-1$ candidate regressors. Each row of $\bfX$ corresponds to one subregion. In addition, $\bfepsilon=(\epsilon_1, \ldots, \epsilon_n)'$ is an $n$-dimensional vector of errors assumed to be independent and identically distributed $N(0,\sigma^2)$. 

The vector $\bfphi=(\phi_1, \phi_2, \dots, \phi_n)'$ contains spatial random effects assumed to be independent of $\bfepsilon$. Specifically, $\bfphi$ follows the sum-zero constrained ICAR specification proposed by \citet{keef:ferr:fran:2018,keef:ferr:fran:2019}. That is, $\bfphi\sim N(\bfzero_n,\tau^{-1}\sigma^2\bfH^+)$, where $\bfH^+$ is the Moore-Penrose inverse of a matrix $\bfH$ that codes the neighborhood structure of the regions under study. The matrix $\bfH$ has elements 
$(H)_{ij} = h_i$ if  $i=j$, 
$(H)_{ij} = -g_{ij}$ if $i \in N_j$, and
$(H)_{ij} = 0$ otherwise, 
where $g_{ij}=g_{ji}\ge 0$ is a measure of similarity between subregions $i$ and $j$. 
Further, $h_i=\sum_{j \neq i}g_{ij}$ is the sum of the off-diagonal elements in row $i$. 
For example, a widely used measure of similarity assigns $g_{ij}=1$ if subregions $i$ and $j$ share a border, and $g_{ij}=0$ otherwise. In this case, $h_i$ is the number of neighbors of subregion $i$. 

This definition implies that $\bfH$ is symmetric and positive semidefinite. Specifically, $n^{-1/2}\boldsymbol{1}$ is a normalized eigenvector of $\bfH$ that corresponds to a null eigenvalue \citep{ferr:deol:2007,deol:2011},  thus $\bfH$ is singular.
We further assume that any two subregions are connected by a path, thus there are no islands and, consequently, the null eigenvalue of $\bfH$ has multiplicity one.  
Because the rank of $\bfH^+$ is $n-1$ with $n^{-1/2}\bfone$ being the normalized eigenvector of $\bfH^+$ corresponding to its null eigenvalue, the sum-zero constrained ICAR distribution $\bfphi \sim N(\bfzero,\tau^{-1} \sigma^2 \bfH^+)$ implicitly encodes the constraint $\bfone' \bfphi = 0.$ 
Then, $\bfphi$ has density \citep{keef:ferr:fran:2018,keef:ferr:fran:2019}
\begin{equation} \label{eqn:ICARdensity}
p(\bfphi) = (2 \pi \sigma^2)^{-(n-1)/2} \tau^{(n-1)/2} \left(\prod_{i=1}^{n-1} d_i \right)^{1/2} \exp\left\{ -\frac{\tau}{2\sigma^2} \bfphi' \bfH \bfphi  \right\} \mathbbm{1}(\bfone' \bfphi = 0),
\end{equation}
where $\tau >0$ is a noise-to-signal ratio parameter and $d_1 \ge \cdots \ge d_{n-1} > d_n=0$ are the ordered eigenvalues of $\bfH$. Importantly, the sum-zero constraint explicitly appears in the term $\mathbbm{1}(\bfone' \bfphi = 0)$ of the density in Equation~(\ref{eqn:ICARdensity}).

There is a fundamental conceptual difference between the sum-zero constrained ICAR model given in Equation~(\ref{eqn:ICARdensity}) ---henceforth referred to as KFF ICAR---  and the improper ICAR model considered by \citet{besag1991} ---henceforth referred to as BYM ICAR. Similarly to the KFF ICAR, the BYM ICAR has density proportional to $\exp\left\{ -\tau \bfphi' \bfH \bfphi/(2\sigma^2)  \right\}$, but without a well-defined constant of proportionality and without an explicit mathematical specification of the sum-zero constraint.
Specifically, the integral of the BYM ICAR density with respect to $\bfphi$ is divergent. In contrast, the integral of the KFF ICAR density is finite and equals one. As a practical consequence, while implementation of Gaussian hierarchical models with BYM ICAR random effects requires sum-zero constraints to be imposed either computationally or through conditioning, implementation of Gaussian hierarchical models with KFF ICAR random effects automatically take into account the sum-zero constraint.

This fundamental conceptual difference between the KFF ICAR and the BYM ICAR has important consequences for the development of statistical methodology. 
For example, to develop a reference prior for the parameters of Gaussian hierarchical models with ICAR random effects, \citet{keef:ferr:fran:2019} had to analytically integrate out the random effects before applying the reference prior methodology. 
It is not clear at all how to perform such analytical integration with BYM ICAR random effects.
Thus, to perform this analytical integration, \citet{keef:ferr:fran:2019} used the KFF ICAR specification. 
As another example, the development of objective Bayes model selection for Gaussian hierarchical models with ICAR random effects requires knowledge of the constant of proportionality. 
Unfortunately, such constant of proportionality is not well-defined in the BYM ICAR specification. 
Thus, \citet{Porter2024} used KFF ICAR random effects to develop objective Bayes model selection for Gaussian hierarchical models. 
Therefore, the KFF ICAR specification has allowed important recent methodological developments. 

Despite the conceptual difference, recent results have shown the equivalence of posterior analyses for Gaussian hierarchical models 
based either on KFF ICAR random effects, or on BYM ICAR random effects with the sum-zero constraint being imposed either computationally or with conditioning.
First, when simulating the random effects without conditioning on data, the sum-zero constrained KFF ICAR is the limiting distribution of a one-at-a-time Gibbs sampler for the BYM ICAR with centering on the fly applied to the simulated $\bfphi$ values at the end of each Gibbs iteration \citep{ferr:2019}. 
In addition, posterior analysis for the Gaussian hierarchical model with sum-zero constrained KFF ICAR random effects is equivalent to using improper BYM ICAR random effects with either centering on the fly \citep{Ferreira2021} or using multivariate Gaussian results to condition on the sum-zero constraint \citep{Porter2024}. 
Therefore, for Gaussian hierarchical models, the three approaches to enforcing the sum-zero constraint for the spatial random effects are equivalent.

\subsection{KFF prior for the model parameters}\label{sec:existing-reference-prior}

Bayesian analysis requires specifying a joint prior for $\bfbeta$, $\sigma^2$, and $\tau$. While informative priors should be utilized whenever valid prior knowledge is available, in an attempt to reflect little or no prior information researchers typically rely on vague priors for $\bfbeta$ and $\sigma^2$. Additionally, even among experts in Bayesian methods, accurately specifying an informative prior for the spatial parameter $\tau$ remains a widespread practical difficulty.

For cases when little or no prior information is available, 
\citet{keef:ferr:fran:2019} developed a reference prior for the parameters 
$(\bfbeta,\sigma^2,\tau)$ of the hierarchical model given by Equation~(\ref{eqn:modelequation}) with ICAR spatial random effects with density (\ref{eqn:ICARdensity}). 
Define $\bfG = \bfI_n - \bfX({\bfX}' \bfX)^{-1} {\bfX}'$ as the orthogonal projection matrix onto the orthogonal complement of the column space of $\bfX$ in $\mathbb{R}^n$. 
In addition, consider the spectral decomposition $\bfG = \bfM\bfL\bfM'$, where $\bfL$ is a diagonal matrix containing the eigenvalues of $\bfG$ sorted in descending order, and the columns of $\bfM$ are the corresponding eigenvectors.
Further, let $\bfM^*$ denote the $n \times (n-p)$ matrix comprising the columns of $\bfM$ that correspond to the nonzero eigenvalues of $\bfG$.
Furthermore, denote by $\lambda_1 \geq \cdots \geq \lambda_{n-p}$ the ordered eigenvalues of ${\bfM^*}'\bfH^+\bfM^*$. 
Then, the reference prior for $(\bfbeta,\sigma^2,\tau)$ developed by \citet{keef:ferr:fran:2019} is
\begin{equation} \label{eqn:referenceprior}
\pi(\alpha,\bfbeta,\sigma^2,\tau) \propto \frac{\pi(\tau)}{\sigma^2},
\end{equation}
where 
\begin{equation} \label{eqn:referenceprior-tau}
\pi(\tau) \propto \frac{1}{\tau} \left[\sum\limits_{j=1}^{n-p}\left(\frac{\lambda_j}{\tau+ \lambda_j}\right)^2 - \frac{1}{n-p}\left \{\sum\limits_{j=1}^{n-p}\left(\frac{\lambda_j}{\tau+ \lambda_j}\right)\right \}^2\right]^{1/2}
\end{equation}
is a marginal density for $\tau$ that is proper \citep{keef:ferr:fran:2019}. In addition, the reference prior density for $(\bfbeta,\sigma^2)$ is $\pi(\bfbeta,\sigma^2)\propto\sigma^{-2}$, which coincides with the reference prior density for the parameters of a linear regression model with independent and identically distributed normal errors. 

The reference prior for $\tau$ provided in Equation~(\ref{eqn:referenceprior-tau}) is of tremendous practical importance for cases when no prior information is available about $\tau$. 
Importantly, \citet{keef:ferr:fran:2019} and \citet{Ferreira2021} have shown that, when compared to widely used gamma priors \citep{best:etal:1999,lee2013}, Bayesian analysis based on the reference prior performs favorably in terms of mean squared estimation error and credible intervals with frequentist coverage close to nominal. These favorable frequentist properties are essential for techniques intended for widespread, automated application. With this in mind, the reference prior developed by \citet{keef:ferr:fran:2019} has been implemented in the ref.ICAR package \citep{port:etal:2025}. This software is freely accessible  via the Comprehensive R Archive Network (CRAN, https://cran.r-project.org/) for the R statistical computing environment \citep{r:2025}.

However, the computation of the marginal prior density $\pi(\tau)$ in (\ref{eqn:referenceprior-tau}) requires the computation of spectral decompositions of two matrices, $\bfG$ and ${\bfM^*}'\bfH^+\bfM^*$, at a cost proportional to $2n^3$ operations. For MCMC implementations, the corresponding eigenvalues may be computed before the start of the MCMC iterations. However, for variable selection these spectral decomposition computations have to be performed for each of the $2^k=2^{p-1}$ possible combinations of candidate regressors. Therefore, the cost of computing the reference prior for all possible competing models is proportional to $ 2^p n^3$. 

In Section~\ref{sec:new-reference-prior}, we propose a novel reference prior that accelerates computations for objective Bayes variable selection. 

\subsection{Objective Bayes model selection}

\citet{Porter2024} have developed an objective Bayes model selection approach for hierarchical models with ICAR random effects. 
Their approach chooses among the candidate regressors and also decides on whether or not spatial random effects should be included in the model. 
Because the reference prior given in Equation (\ref{eqn:referenceprior}) is improper (for $\bfbeta$ and $\sigma^2$), \citet{Porter2024} developed a model selection approach based on fractional Bayes factors (FBF) \citep{ohagan1995}. \citet{Porter2024} compared this FBF approach versus the widely applicable information criterion (WAIC) \citep{wata:2010}, the deviance information criterion (DIC) \citep{spiegelhalter2002}, and the type II DIC \citep{cele:etal:2006}. Compared to these three competing approaches, the FBF approach provides much better performance in terms of both identifying the correct set of regressors and the correct spatial dependence structure.

First, \citet{Porter2024} assign equal prior probability $1/2$ for spatial dependence and for independence. In addition, \citet{Porter2024} divide the prior probabilities as suggested by \cite{scott2010} so that the full set of candidate models contains all possible combinations of regressors under both spatial dependence and independence. Consider a model $M_c$ (with or without spatial dependence) with $k_c=p_c-1$ regressors in an $n\times p_c$ regression matrix $\bfX_c$. 
In addition, recall $k=p-1$ is the total number of candidate regressors.
Then, \citet{Porter2024} assign for $M_c$ the prior probability
\begin{equation} \label{eq:model_priors}
    P(M_c) = \frac{1}{2(k+1)} \binom{k}{k_c}^{-1}.
\end{equation}

The FBF approach trains the reference prior with a fraction $b=m/n$ of the likelihood function, where \citet{Porter2024} set $m$ to be the minimum training sample size $m=p+1$. 
Then, the fractional integrated likelihood under model $M_c$ is 
\begin{align} 
    q_c(b,\bfy) & = \int \frac{\pi(\boldsymbol{\eta}_c)p^b(\bfy|\boldsymbol{\eta}_c,M_c)\{p^{1-b}(\bfy|\boldsymbol{\eta}_c,M_c)\}d\boldsymbol{\eta}_c}{\int p^b(\bfy|\boldsymbol{\eta}_c,M_c)\pi(\boldsymbol{\eta}_c)d\boldsymbol{\eta}_c} = \frac{\int p(\bfy|\boldsymbol{\eta}_c,M_c)\pi(\boldsymbol{\eta}_c)d\boldsymbol{\eta}_c}{\int p^b(\bfy|\boldsymbol{\eta}_c,M_c)\pi(\boldsymbol{\eta}_c)d\boldsymbol{\eta}_c}, \label{eq:12_int_likelihood_fbf}
\end{align}
where the integral in the denominator of Equation (\ref{eq:12_int_likelihood_fbf}) is
\begin{align} \label{eq:13_int_likelihood_ref}
   \int p^b(\bfy|\boldsymbol{\eta}_{c},M_{c})\pi(\boldsymbol{\eta}_{c})d\boldsymbol{\eta}_{c} \propto \int_0^{\infty} & |\bfOmega|^{-\frac{b}{2}}|\bfX_{c}' \bfOmega^{-1} \bfX_{c}|^{-\frac{1}{2}} \Big[\frac{b}{2}S_{c}^2\Big]^{\frac{p_{c}-nb}{2}} \pi(\tau) d\tau, 
\end{align}
\noindent with $\bfOmega = \bfI_n + \tau^{-1} \bfH^+$ and $S_{c}^2=\bfy'(\bfOmega^{-1}-\bfOmega^{-1}\bfX_{c}(\bfX_{c}' \bfOmega^{-1}\bfX_{c})^{-1}\bfX_{c}'\bfOmega^{-1})\bfy$. This is a one-dimensional integral that may be accurately approximated using adaptive quadrature integration. 
To increase the speed and accuracy of the adaptive quadrature integration, we reparameterize from $\tau$ to $\psi=\log(\tau)$, with the prior density of $\psi$ equal to $\pi_\psi(\psi) = e^\psi \pi(\tau)$. Then, the integrand becomes $h(\psi)=|\bfOmega|^{-\frac{b}{2}}|\bfX_{c}' \bfOmega^{-1} \bfX_{c}|^{-1/2} (0.5 b S_{c}^2)^{(p_{c}-nb)/2} \pi_\psi(\psi)$, with limits of integration $-\infty$ and $\infty$. 
While direct computations of $|\bfOmega|$, $|\bfX_{c}' \bfOmega^{-1} \bfX_{c}|$, and $S_{c}^2$ scale as $O(n^3)$, we propose in Section~\ref{sec:spectral-domain-computations} to compute these quantities with spectral computations that scale as $O(n)$.

The fractional Bayes factor of model $M_c$ versus model $M_a$ is $BF_{ca}^b = q_c(b,\bfy)/q_a(b,\bfy)$. 
In addition, posterior probabilities of the several competing models can be computed from the Bayes factors \citep{Liang2008,Porter2024}. 
Finally, these posterior model probabilities can be used to compute posterior inclusion probabilities for the candidate regressors \citep{Barbieri2004}.

\section{Spectral domain computations for objective Bayes model selection}\label{sec:spectral-domain-computations}

\citet{Ferreira2021} has proposed to accelerate computations for Gaussian hierarchical models with ICAR spatial random effects by transforming the hierarchical model from the spatial domain
to the spectral domain. 
This section adapts the proposal by \citet{Ferreira2021} to objective Bayes model selection. 

Let $\bfH = \bfP\bfD\bfP'$ be the spectral decomposition of $\bfH$, where the columns of $\bfP=(\bfp_1,\bfp_2,\dots,\bfp_n)$ are the normalized eigenvectors of $\bfH$. 
In addition, the diagonal matrix $\bfD = \diag(d_1, d_2, \dots , d_n)$ contains the ordered eigenvalues of $\bfH$ denoted by $d_1 \ge d_2 \ge \cdots \ge d_{n-1} > d_n = 0$.  
Denote the Moore-Penrose inverse of $\bfD$ by $\bfD^+=\diag(d_1^{-1},\ldots,d_{n-1}^{-1},0)$.
Now, premultiply Equation~(\ref{eqn:modelequation}) by the transpose of the eigenvectors matrix $\bfP$. 
Denote the spectrally transformed quantities by $\ysp = \bfP'\bfy,$ $\xsp = \bfP'\bfX,$ $\bfxi = \bfP'\bfphi,$ and $\bfzeta=\bfP'\bfepsilon$. 
Because $Cov(\bfzeta) = Cov(\bfP'\bfepsilon) = \sigma^2\bfP'\bfP = \sigma^2\bfI$, the unstructured errors in the spectral domain have the same distribution as in the original domain, that is $\bfzeta \sim N(\bfzero,\sigma^2 \bfI)$.
In addition, the covariance matrix of the spatial random effects in the spectral domain is $Cov(\bfxi) = Cov(\bfP'\bfphi) = \sigma^2\tau^{-1}\bfP'\bfH^+\bfP = \sigma^2\tau^{-1}\bfD^{+} = \sigma^2\tau^{-1} \diag(d_1^{-1},\ldots,d_{n-1}^{-1},0)$.
Thus, in the spectral domain the hierarchical model with ICAR spatial random effects may be written as \citep{Ferreira2021}
\begin{eqnarray}
\ysp & = & \xsp \bfbeta + \bfxi + \bfzeta,  \ \ \bfzeta \sim N(\bfzero,\sigma^2 \bfI), \label{eqn:spectral.hier.model1} \\
\bfxi  & | & \sigma^2, \tau  \sim  N(\bfzero,\sigma^2 \tau^{-1} \bfD^+). \label{eqn:spectral.hier.model2}
\end{eqnarray}

Now, integrate out the random effects $\bfxi$ to obtain in the spectral domain  the model
\begin{equation}
\ysp = \xsp \bfbeta + \bfeta,  \ \ \bfeta \sim N(\bfzero,\sigma^2 \{\bfI + \tau^{-1} \bfD^+\}). \label{eqn:spectral.integrated likelihood} 
\end{equation}

Because the covariance matrix $\bfI + \tau^{-1} \bfD^+$ is diagonal, computations in the spectral domain are much simpler, faster, and more scalable than computations in the original spatial domain. For example, since $\bfD^+=\diag(d_1^{-1},\ldots,d_{n-1}^{-1},0)$, then
$|\bfOmega| = |\bfI_n + \tau^{-1} \bfH^+| = |\bfI_n + \tau^{-1} \bfP\bfD^+\bfP'| = |\bfI_n + \tau^{-1} \bfD^+| = \prod_{i=1}^{n-1}(1+\tau^{-1}d_i^{-1})$ which costs $O(n)$ operations.

Let $\bfB(\tau)=(\bfI + \tau^{-1} \bfD^+)^{-1}$. In addition, let 
\begin{equation} \label{eqn:definition-of-bfb}
\bfb(\tau) = \left( \frac{\tau d_1}{\tau d_1 + 1}, \cdots, \frac{\tau d_{n-1}}{\tau d_{n-1} + 1}, 1 \right)'.
\end{equation}
Then, $\bfB(\tau)=\diag(\bfb(\tau))$.
Further, let $\otimes$ denote the Kronecker product and $\odot$ denote the Hadamard product \citep[p. 45,][]{magn:neud:1999}
that returns the matrix of element-wise products. 
Let $\bfQ(\tau)=
\xsp' (\bfI + \tau^{-1} \bfD^+)^{-1} = \xsp'\bfB(\tau) = 
\left\{ \xsp \odot (\bfone_p' \otimes \bfb(\tau)) \right\}',
$
where the operation on the left costs about $O(n^2)$ operations, whereas the operation on the right costs just $O(n)$ operations. 
Then, 
$|\bfX_{c}' \bfOmega^{-1} \bfX_{c}| 
= |\bfX_{c}' (\bfI_n + \tau^{-1} \bfH^+)^{-1} \bfX_{c}|
=|\bfX_{c}' (\bfI_n + \tau^{-1} \bfP\bfD^+\bfP')^{-1} \bfX_{c}|
=|\bfX_{c}' \bfP(\bfI_n + \tau^{-1} \bfD^+)^{-1} \bfP'\bfX_{c}|
=|\xsp_{c}' (\bfI_n + \tau^{-1} \bfD^+)^{-1} \xsp_{c}|
=|\bfQ(\tau)\xsp_{c}|$. Thus, in terms of sample size, this determinant costs $O(n)$ operations. 

Finally, let $\bfR(\tau)=
\ysp' (\bfI + \tau^{-1} \bfD^+)^{-1} = \ysp'\bfB(\tau) = 
\left\{ \ysp \odot \bfb(\tau) \right\}'
$. Then
\begin{eqnarray}
S_{c}^2 & = & \bfy'(\bfOmega^{-1}-\bfOmega^{-1}\bfX_{c}(\bfX_{c}' \bfOmega^{-1}\bfX_{c})^{-1}\bfX_{c}'\bfOmega^{-1})\bfy \nonumber \\
& = & \ysp'(\bfB(\tau)-\bfB(\tau)\xsp_{c}(\bfQ(\tau)\xsp_{c})^{-1}\xsp_{c}'\bfB(\tau))\ysp \nonumber \\
& = & \bfR(\tau)\ysp-\bfR(\tau)\xsp_{c}(\bfQ(\tau)\xsp_{c})^{-1}\xsp_{c}'\bfR'(\tau), \nonumber
\end{eqnarray}
which costs $O(n)$ operations for each model $M_c$.

\section{New reference prior for the model parameters}\label{sec:new-reference-prior}

The reference prior presented in Section~\ref{sec:existing-reference-prior} was developed by \citet{keef:ferr:fran:2019} based on results given by \citet{deoliveira2007}. In this section, we provide another reference prior that is based on a theorem given by \citet{berg:oliv:sans:2001}; for completeness, we provide their theorem in the Auxiliary Facts section of the Appendix.

Recall that $\bfB(\tau) = (\bfI_{n} + \tau^{-1}\bfD^{+})^{-1}$. 
Let $\bfQ_{ij}(\tau) = \widetilde{\bfX}' \{\bfD^{+}\}^{i} \{\bfB(\tau)\}^{j} \widetilde{\bfX}$. 
Because both $\bfB(\tau)$ and $\bfD^{+}$ are diagonal, $\bfQ_{ij}(\tau)=\widetilde{\bfX}' \{\bfB(\tau)\}^{j} \{\bfD^{+}\}^{i} \widetilde{\bfX}$. 
Then, the following theorem provides a form for the reference prior of $(\bfbeta,\sigma^2,\tau)$ that is based on traces of matrices.

\begin{theorem}\label{theo:compact-form-reference-prior} The reference prior of $(\bfbeta,\sigma^2,\tau)$ is of the form $\pi(\bfbeta,\sigma^2,\tau) \propto \sigma^{-2}\pi(\tau)$, where
\begin{eqnarray}\label{eqn:compact-form-reference-prior}
\pi(\tau) & \propto & \tau^{-2} \bigg\{\text{tr}\left((\bfD^+ \bfB(\tau))^2\right) 
    + \text{tr}\left( (\{\bfQ_{01}(\tau)\}^{-1} \bfQ_{12}(\tau))^2\right) 
  - 2\text{tr}\left( \{\bfQ_{01}(\tau)\}^{-1} \bfQ_{23}(\tau) \right) \nonumber \\
& & \ \ \ \ \ \ \ \ -\frac{1}{n-q} \left[\text{tr}(\{\bfQ_{01}(\tau)\}^{-1}\bfQ_{12}(\tau))-\text{tr}(\bfD^{+}\bfB(\tau))\right]^2 \bigg\}^{1/2}.
\end{eqnarray}
\end{theorem}

\begin{proof}
See the Appendix.
\end{proof}

The largest $n-1$ eigenvalues $d_1, \ldots, d_{n-1}$ of the matrix $\bfH$ (recall $d_n=0$) can be computed just once and used for all competing models. These eigenvalues are used to compute the matrices $\bfD^+$, $\bfB(\tau)$, and $\bfQ_{ij}(\tau)$. Given these eigenvalues, the computation of the prior given in Equation~(\ref{eqn:compact-form-reference-prior}) scales as $O(n)$.


%

\begin{theorem}\label{theo:equivalence-of-priors}
The marginal reference prior for $\tau$ given in Equation~(\ref{eqn:compact-form-reference-prior}) is equivalent to the marginal reference prior for $\tau$ given in Equation~(\ref{eqn:referenceprior-tau}).
\end{theorem}

\begin{proof}
See the Appendix.
\end{proof}

Theorem~\ref{theo:equivalence-of-priors} establishes that the marginal reference prior for $\tau$ given in Equation~(\ref{eqn:compact-form-reference-prior}) is a remarkably elegant formulation that transforms the sum over eigenvalues used in Equation~(\ref{eqn:referenceprior-tau}) into a set of trace operations that entirely bypasses the need for an explicit eigendecomposition of the projection space. As we show in Section~\ref{sec:computational-time}, this novel formulation has tremendous practical impact in terms of computational time savings.

\begin{figure}[t!]
\begin{center} 
\includegraphics*[width=3in,height=3in]{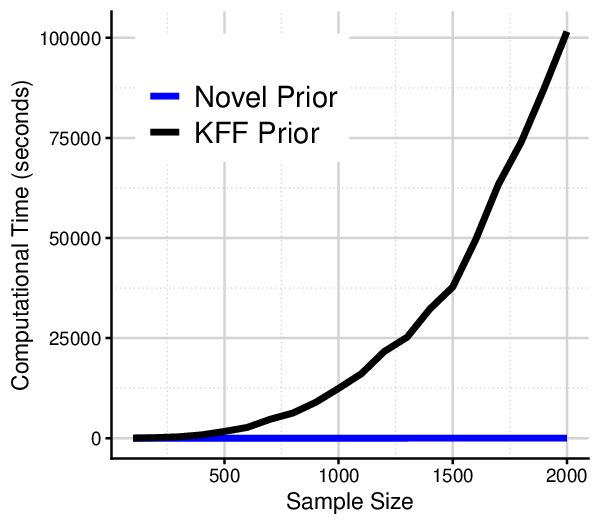} 
\caption{\label{fig:compute-time-model-selection} Computational times in seconds as a function of sample size for Bayesian model selection based on exhaustive search for 5 regressors.}
\end{center}
\end{figure}

\section{Simulation study}\label{sec:computational-time}

This section presents a simulation study to compare the computational cost  of the KFF prior and our novel reference prior. Computations for the KFF prior are as implemented in the R package ref.ICAR version 2.0.2 \citep{port:etal:2025}.
All computations presented here were performed in R version 4.2.1 optimized with Intel's Math Kernel Library on a MacBook Pro with a 2.7 GHz Intel Core i7 processor and  MacOS Monterey operating system. 

We simulate data from hierarchical models with a one-dimensional spatial dataset with ICAR random effects. In this simulation study, we consider  sample sizes: $100, 200, \ldots, 2000$. In addition, the variance parameters are $\tau=0.3$ and $\sigma^2=1$. The regressors were simulated from i.i.d. standard normal distributions, and then centered to sum to zero. We consider one intercept and 5 candidate regressors. The vector of regression coefficients including the intercept is $\bfbeta=(1, 1,1,0,0,0)'$. 

Figure~\ref{fig:compute-time-model-selection} shows the computational time in seconds for computations based on the KFF prior and on our novel reference prior as a function of sample size. For all datasets, computations based our new reference prior were faster than those based on the KFF prior. Specifically, for a sample size of 100 regions, computations for the KFF prior took 18.8 seconds and those based on the novel reference prior took 1 second. More dramatically, for a sample size of 2000 regions, computations for the KFF prior took 28 hours whereas those based on the novel reference prior took only 19.8 seconds.
Therefore, when compared to the KFF prior, our novel reference prior produces a substantial decrease in computational time.

\section{Application: Household income in the United States\label{sec:application}}

To illustrate the utility of our novel reference prior for model selection for large spatial datasets, we consider a hierarchical model with ICAR random effects for the logarithm of median household income in the contiguous United States in 2017 per county for a total of 3,108 counties (or similar geopolitical entities). These data have been previously analyzed by \citet{Ferreira2021} and \citet{Porter2024}. When compared to these previous works, here we consider additional regressors. 
Specifically, we consider the following 11 county-level socio-economic regressors: 
the logarithm of population size, 
the logarithm of the percentage of adults with high school degree, 
the logarithm of the percentage of adults with some college or associate degree,
the logarithm of the percentage of adults with bachelors degree or higher, and indicators of whether the county was in a large metropolitan area, medium metropolitan area, small metropolitan area, urban area adjacent to metropolitan area, urban area not adjacent to metropolitan area, and rural area adjacent to metropolitan area. Hence, in the model considered, the baseline education is less than high school and the baseline metro status is rural area not adjacent to a metropolitan area.

Computations for this application were performed on the same laptop used for the computations presented in Section~\ref{sec:computational-time}.
Computations using the KFF prior would have taken several months and, thus, were deemed not practically feasible. Hence, here we report results based only on our novel reference prior. Computations based on our novel reference prior took 27.3 minutes. 
Therefore, the novel reference prior we propose leads to cost-effective computations.

\begin{table}[t]
\centering
\caption{Case study: Logarithm of median household income. 
Posterior inclusion probabilities for candidate regressors based on novel reference prior. 
Baseline education is less than high school. 
Baseline metro status is rural area not adjacent to a metropolitan area. 
\label{table:posterior-inclusion-probabilities}}
\begin{tabular}{c c}
\hline
\textbf{Regressor} & \textbf{Posterior inclusion probability} \\
\hline
log(population) & 0.3853  \\
log(high school) & 0.1403  \\
log(associate degree) & 0.9998  \\
log(bachelor) & 1.0000  \\
log(unemployment) & 1.0000  \\
metro large & 1.0000  \\
metro medium & 1.0000  \\
metro small & 1.0000  \\
urban adjacent & 1.0000  \\
urban not adjacent & 0.9988  \\
rural adjacent & 1.0000  \\
\hline
\end{tabular}
\end{table}

Table~\ref{table:posterior-inclusion-probabilities} presents the 
posterior inclusion probabilities for each of the 11 candidate regressors based on the novel reference prior. The posterior inclusion probabilities of each of the county metro status indicators are tremendously close to one. Hence, the county metro status is an important predictor of county level median household income. In addition, the posterior inclusion probability of the logarithm of population size is 0.3853, which indicates that when including county metro status in the model, the data does not support the logarithm of population size as an important predictor of county level median household income. Further, given its posterior inclusion probability of 0.1403, and keeping in mind that the the baseline education is less than high school, there is some evidence that the logarithm of the percentage of adults with high school degree is not an important predictor of median household income. Finally, there is strong evidence that the percentage of adults with associate degree or some college (pip=0.9998) and the percentage of adults with bachelors degree (pip=1.0000) are important predictors of county level median household income.

\section{Discussion} \label{sec:discussion}

We have proposed a novel reference prior for Gaussian hierarchical models with ICAR random effects. 
When compared to a previously proposed reference prior \citep{keef:ferr:fran:2019}, this novel prior provides equivalent posterior analysis and, for large sample sizes, is orders of magnitude faster for objective Bayes model selection computations.

Our results are for hierarchical models with Gaussian observations. 
Because there are no reference priors for hierarchical models with ICAR random effects when the observations are nonGaussian, for example binomial or Poisson, a promising future research direction is the development of reference priors for these models. 
Another related possible future direction is the development of spectral computations for such models, which would be useful for the analysis of spatial nonGaussian data.

\section*{Declaration of generative AI and AI-assisted technologies in the manuscript preparation process}

During the preparation of this work the author(s) used Gemini 3.1 Pro in order to help develop the proof of Theorem~\ref{theo:equivalence-of-priors}. After using this tool/service, the author reviewed and edited the content as needed and takes full responsibility for the content of the published article.

\section*{Appendix}

\begin{Appendix}
\setcounter{section}{0}
\renewcommand\thesection{\Alph{section}}

\section*{Appendix 1: Auxiliary results}

\begin{theorem}[\citealp{berg:oliv:sans:2001}]
\label{Theo:Bergeretal2001}
Consider the model $\bfY \mid Q \sim N(\bfZ\bftheta, \sigma^{2}\bfSigma_{\tau})$,
where $\bfY$ is an $n$-dimensional vector of observations, $\bfZ$ is an $n$ by $p$ matrix of regressors, $\bftheta$ is a $p$-dimensional vector of regression coefficients, $\sigma^{2}>0$, and $\bfSigma_{\tau}$ is a positive definite matrix that depends on a parameter $\tau$. 
Define the matrix $\bfW_{\tau}$ as
\begin{equation}
    \bfW_{\tau} = \left(\frac{\partial}{\partial \tau}\bfSigma_{\tau}\right)\bfSigma_{\tau}^{-1}\bfP_\tau^{\Sigma}
\end{equation}
where $\bfP_\tau^{\Sigma} = \bfI_n - \bfZ(\bfZ'\bfSigma_{\tau}^{-1}\bfZ)^{-1}\bfZ'\bfSigma_{\tau}^{-1}$ is a projection matrix. 
Then, the reference prior for $(\bftheta, \sigma^2, \tau)$ is of the form $\pi(\bftheta, \sigma^2, \tau) \propto \sigma^{-2}\pi(\tau)$, where the marginal reference prior for $\tau$ is 
\begin{equation}
    \pi(\tau) \propto \left\{\text{tr}(\bfW_{\tau}^{2}) - \frac{1}{n-p}(\text{tr}(\bfW_{\tau}))^{2}\right\}^{1/2}.
\end{equation}
\end{theorem}


\section*{Appendix 2: Proof of Theorem \ref{theo:compact-form-reference-prior}}

The model we consider is $\bfy \sim N(\bfX\bfbeta,\sigma^2(\bfI_n+\tau^{-1}\bfH^+))$. 
In the notation of Theorem~\ref{Theo:Bergeretal2001}, $\bfY=\bfy$, $\bfZ = \bfX$, $\bftheta = \bfbeta$, and
$\bfSigma_{\tau} = \bfI_{n} + \tau^{-1}\bfH^{+}$.

Consider the spectral decomposition $\bfH = \bfP\bfD\bfP'$, where $\bfP=(\bfp_1, \ldots, \bfp_n)$ is the matrix of eigenvectors of $\bfH$, $\bfp_n=n^{-1/2}\bfone_n$, $\bfD=\diag(d_1,\ldots,d_n)$, and $d_1\geq\cdots\geq d_{n-1} >d_n=0$ are the eigenvalues of $\bfH$. 
Then, $\bfSigma_{\tau} = \bfP(\bfI_{n} + \tau^{-1}\bfD^{+})\bfP'$, its inverse is
$\bfSigma_{\tau}^{-1} = \bfP(\bfI_{n} + \tau^{-1}\bfD^{+})^{-1}\bfP'$,
and its derivative with respect to $\tau$ is:
\begin{equation}
    \frac{\partial}{\partial\tau}\bfSigma_{\tau} = -\tau^{-2}\bfH^{+} = -\tau^{-2}\bfP\bfD^{+}\bfP',
\end{equation}
where $\bfD^+=\diag(d_1^{-1},\ldots,d_{n-1}^{-1},0)$ is the Moore-Penrose inverse of $\bfD$, and $\bfH^{+} =\bfP\bfD^{+}\bfP'$ is the Moore-Penrose inverse of $\bfH$.

Substituting the components into the definition of $\bfW_{\tau}$, we obtain
\begin{align}
    \bfW_{\tau} &= \left(\frac{\partial}{\partial\tau}\bfSigma_{\tau}\right)\bfSigma_{\tau}^{-1}\bfP_{\tau}^{\Sigma} \nonumber \\
    &= -\tau^{-2}\bfP\bfD^{+}\bfP' \bfP(\bfI_{n} + \tau^{-1}\bfD^{+})^{-1} \bfP' \bfP_{\tau}^{\Sigma} \nonumber \\
    &= -\tau^{-2}\bfP\bfD^{+}(\bfI_n + \tau^{-1}\bfD^{+})^{-1}\bfP' \bfP_{\tau}^{\Sigma} \nonumber \\
    &= -\tau^{-2}\bfP\bfD^{+}(\bfI_n + \tau^{-1}\bfD^{+})^{-1}\bfP' \nonumber \\
    &\quad + \tau^{-2}\bfP\bfD^{+}(\bfI_n + \tau^{-1}\bfD^{+})^{-1} \bfP' \bfX(\bfX'\bfP(\bfI_n + \tau^{-1}\bfD^{+})^{-1} \bfP'\bfX)^{-1} \bfX'\bfP(\bfI_n + \tau^{-1}\bfD^{+})^{-1}\bfP'.
\end{align}

Let   $\bfB(\tau) = (\bfI_{n} + \tau^{-1}\bfD^{+})^{-1}$ and recall that 
$\widetilde{\bfX}=\bfP'\bfX$. Then, 
\begin{align}
    \bfW_{\tau} &= -\tau^{-2}\bfP\bfD^{+}\bfB(\tau)\bfP' \nonumber \\
    &\quad + \tau^{-2}\bfP\bfD^{+}\bfB(\tau)  \widetilde{\bfX}(\widetilde{\bfX}'\bfB(\tau) \widetilde{\bfX})^{-1} \widetilde{\bfX}'\bfB(\tau)\bfP'.
\end{align}

Let $\bfQ_{ij}(\tau) = \widetilde{\bfX}' \{\bfD^{+}\}^{i} \{\bfB(\tau)\}^{j} \widetilde{\bfX}$. Because both $\bfB(\tau)$ and $\bfD^{+}$ are diagonal, $\bfQ_{ij}(\tau)=\widetilde{\bfX}' \{\bfB(\tau)\}^{j} \{\bfD^{+}\}^{i} \widetilde{\bfX}$. Because of the linearity and cyclic properties of the trace, we obtain
\begin{align}
    \text{tr}(\bfW_{\tau}) &= -\tau^{-2}\text{tr}(\bfP\bfD^{+}\bfB(\tau)\bfP')\nonumber \\
    &\quad + \tau^{-2}\text{tr}(\bfP\bfD^{+}\bfB(\tau)  \widetilde{\bfX}(\widetilde{\bfX}'\bfB(\tau) \widetilde{\bfX})^{-1} \widetilde{\bfX}'\bfB(\tau)\bfP') \nonumber\\
 & = -\tau^{-2}\text{tr}(\bfD^{+}\bfB(\tau)) + \tau^{-2}\text{tr}(\{\bfQ_{01}(\tau)\}^{-1}\bfQ_{12}(\tau)).   \label{eqn:traceW}
\end{align}

In addition,
\begin{align}
    \bfW_{\tau}^2 &= \tau^{-4} \bfP \bfD^+ \bfB(\tau) \bfD^+ \bfB(\tau) \bfP' \nonumber\\
&\quad + \tau^{-4} \bfP \bfD^+ \bfB(\tau) \widetilde{\bfX}(\widetilde{\bfX}'\bfB(\tau) \widetilde{\bfX})^{-1} \widetilde{\bfX}'\bfB(\tau) \bfD^+ \bfB(\tau) \widetilde{\bfX} (\widetilde{\bfX}'\bfB(\tau) \widetilde{\bfX})^{-1} \widetilde{\bfX}'\bfB(\tau) \bfP' \nonumber\\
&\quad - 2 \tau^{-4} \bfP \bfD^+ \bfB(\tau) \bfD^+ \bfB(\tau) 
\widetilde{\bfX} (\widetilde{\bfX}'\bfB(\tau) \widetilde{\bfX})^{-1} \widetilde{\bfX}'\bfB(\tau) \bfP'. \nonumber
\end{align}

Using again the linearity and circularity properties of the trace, we obtain
\begin{align}
    \text{tr}(\bfW_{\tau}^2) &= \tau^{-4} \text{tr}\left(\bfD^+ \bfB(\tau) \bfD^+ \bfB(\tau)\right) \nonumber\\
&\quad + \tau^{-4} \text{tr}\left(\bfD^+ \bfB(\tau) \widetilde{\bfX}(\widetilde{\bfX}'\bfB(\tau) \widetilde{\bfX})^{-1} \widetilde{\bfX}'\bfB(\tau) \bfD^+ \bfB(\tau) \widetilde{\bfX} (\widetilde{\bfX}'\bfB(\tau) \widetilde{\bfX})^{-1} \widetilde{\bfX}'\bfB(\tau)\right) \nonumber\\
&\quad - 2 \tau^{-4} \text{tr}\left(\bfD^+ \bfB(\tau) \bfD^+ \bfB(\tau) 
\widetilde{\bfX} (\widetilde{\bfX}'\bfB(\tau) \widetilde{\bfX})^{-1} \widetilde{\bfX}'\bfB(\tau)\right) \nonumber\\
&= \tau^{-4} \left[\text{tr}\left((\bfD^+ \bfB(\tau))^2\right) 
    + \text{tr}\left( (\{\bfQ_{01}(\tau)\}^{-1} \bfQ_{12}(\tau))^2\right) 
  - 2\text{tr}\left( \{\bfQ_{01}(\tau)\}^{-1} \bfQ_{23}(\tau) \right)\right]. \label{eqn:traceW2}
\end{align}
Then, Equations~(\ref{eqn:traceW}) and (\ref{eqn:traceW2}) together with Theorem~\ref{Theo:Bergeretal2001} 
provide the reference prior for $(\alpha,\bfbeta,\sigma^2,\tau)$ given in Theorem \ref{theo:compact-form-reference-prior}.
\hfill $\square$

\section*{Appendix 3: Proof of Theorem \ref{theo:equivalence-of-priors}}

Let 
$$
T_1 = \sum_{j=1}^{n-p} \frac{\lambda_j}{\tau + \lambda_j} 
\mbox{ \ \ \ and \ \ \ }
T_2 = \sum_{j=1}^{n-p} \left(\frac{\lambda_j}{\tau + \lambda_j}\right)^2,
$$
where $\lambda_1, \dots, \lambda_{n-p}$ are the eigenvalues of the matrix ${\mathbf{M}^*}'\mathbf{H}^+\mathbf{M}^*$. Then, the marginal prior for $\tau$ given in Equation~(\ref{eqn:referenceprior-tau}) can be written as
\begin{equation}\label{eqn:referenceprior-tau-T1-T2}
\pi(\tau) \propto \frac{1}{\tau} \left[T_2 - \frac{1}{n-p}T_1^2\right]^{1/2}.
\end{equation}

Because $\mathbf{G} = \mathbf{M}\mathbf{L}\mathbf{M}'$ and $\mathbf{G}$ is an idempotent projection matrix with rank $n-p$, the $n \times n$ matrix $\mathbf{G}\mathbf{H}^+\mathbf{G}$ has $p$ eigenvalues equal to zero and $n-p$ non-zero eigenvalues equal to $\lambda_1, \dots, \lambda_{n-p}$. 
Let $\mathbf{C}_W = \widetilde{\bfG} \mathbf{D}^+ \widetilde{\bfG}$, where $\widetilde{\bfG} = \bfP'\bfG\bfP = \mathbf{I}_n - \widetilde{\bfX}(\widetilde{\bfX}'\widetilde{\bfX})^{-1}\widetilde{\bfX}'$. Because $\mathbf{C}_W$ is similar to $\mathbf{G}\mathbf{H}^+\mathbf{G}$, $\mathbf{C}_W$ also has $p$ eigenvalues equal to zero and $n-p$ non-zero eigenvalues equal to $\lambda_1, \dots, \lambda_{n-p}$. 
Let $\mathbf{E} = (\mathbf{D}^+)^{1/2}$. 
Because of the idempotence of the projection matrix $\widetilde{\bfG}$, the cyclic permutation of eigenvalues, and the fact that the null space of $\mathbf{E}$ is entirely contained in the null space of $\widetilde{\mathbf{G}}$, 
the matrix $\mathbf{E}\widetilde{\bfG}\mathbf{E}$ also has $p$ eigenvalues equal to zero and $n-p$ non-zero eigenvalues equal to $\lambda_1, \dots, \lambda_{n-p}$. 

Let $\mathbf{V} = (\tau \mathbf{I}_n + \mathbf{E}\widetilde{\bfG}\mathbf{E})^{-1}$. Thus, $\mathbf{V}$ has $n-p$ eigenvalues equal to $(\tau + \lambda_j)^{-1}$ and $p$ eigenvalues equal to $\tau^{-1}$. As a consequence, the trace of $\bfV$ is $\sum (\tau + \lambda_j)^{-1} + p\tau^{-1}$. 

We can expand $\mathbf{V}$ as
$$\mathbf{V} = (\tau \mathbf{I}_n + \mathbf{D}^+ - \mathbf{E}\widetilde{\bfX}(\widetilde{\bfX}'\widetilde{\bfX})^{-1}\widetilde{\bfX}'\mathbf{E})^{-1}.$$

Let $\mathbf{R} = \tau \mathbf{I}_n + \mathbf{D}^+$. 
By applying the Woodbury matrix identity to $\mathbf{V}$, we get
$$\mathbf{V} = \mathbf{R}^{-1} + \mathbf{R}^{-1}\mathbf{E}\widetilde{\bfX} (\widetilde{\bfX}'\widetilde{\bfX} - \widetilde{\bfX}'\mathbf{E}\mathbf{R}^{-1}\mathbf{E}\widetilde{\bfX})^{-1} \widetilde{\bfX}'\mathbf{E}\mathbf{R}^{-1}.$$

Recall that $\bfB(\tau) = (\mathbf{I}_n + \tau^{-1}\mathbf{D}^+)^{-1}$.
Thus, $\mathbf{R}^{-1} = \tau^{-1}\bfB(\tau)$. 
Hence, $\mathbf{E}\mathbf{R}^{-1}\mathbf{E} = \tau^{-1}\mathbf{D}^+\bfB(\tau)$.
In addition, recall that $\bfQ_{ij}(\tau)=\widetilde{\bfX}' \{\bfD^{+}\}^{i} \{\bfB(\tau)\}^{j}  \widetilde{\bfX}$. 
Then, the inverted term becomes
$$(\widetilde{\bfX}'\widetilde{\bfX} - \widetilde{\bfX}'\mathbf{E}\mathbf{R}^{-1}\mathbf{E}\widetilde{\bfX})^{-1}
=
\widetilde{\bfX}'(\mathbf{I}_n - \tau^{-1}\mathbf{D}^+\bfB(\tau))\widetilde{\bfX} = \widetilde{\bfX}'\bfB(\tau)\widetilde{\bfX} = \mathbf{Q}_{01}(\tau).$$

Substituting this back in the expression for $\mathbf{V}$, we get
\begin{equation}\label{eqn:master-equation-for-V}
\mathbf{V} = \tau^{-1}\bfB(\tau) + \tau^{-1}\bfB(\tau)\mathbf{E}\widetilde{\bfX} \{\mathbf{Q}_{01}(\tau)\}^{-1} \widetilde{\bfX}'\mathbf{E} \tau^{-1}\bfB(\tau).
\end{equation}

We can rewrite $T_1$ as
$$T_1 = \sum_{j=1}^{n-p} \frac{\lambda_j}{\tau + \lambda_j} = \sum_{j=1}^{n-p} \left(1 - \frac{\tau}{\tau + \lambda_j}\right) = (n-p) - \tau \sum_{j=1}^{n-p} (\tau + \lambda_j)^{-1}.$$

Because the trace of $\mathbf{V}$ is $\sum (\tau + \lambda_j)^{-1} + p\tau^{-1}$, 
$$T_1 = n - \tau\text{tr}(\mathbf{V})$$

Taking the trace of $\mathbf{V}$ in Equation~(\ref{eqn:master-equation-for-V}) and applying cyclic trace properties, we get
\begin{equation}\label{eqn:trace-V}
\tau\text{tr}(\mathbf{V}) = \text{tr}(\bfB(\tau)) + \tau^{-1}\text{tr}(\{\mathbf{Q}_{01}(\tau)\}^{-1}\mathbf{Q}_{12}(\tau)).
\end{equation}

Thus, $T_1 = \text{tr}(\mathbf{I}_n - \bfB(\tau)) - \tau^{-1}\text{tr}(\{\mathbf{Q}_{01}(\tau)\}^{-1}\mathbf{Q}_{12}(\tau))$.
Note $\mathbf{I}_n - \bfB(\tau) = \tau^{-1}\mathbf{D}^+\bfB(\tau)$, hence
\begin{equation}\label{eqn:T1-final}
T_1 = \tau^{-1} \left[ \text{tr}(\mathbf{D}^+\bfB(\tau)) - \text{tr}(\{\mathbf{Q}_{01}(\tau)\}^{-1}\mathbf{Q}_{12}(\tau)) \right].
\end{equation}

In addition, we can rewrite $T_2$ as
\begin{equation}\label{eqn:T2}
T_2 = \sum_{j=1}^{n-p} \left(\frac{\lambda_j}{\tau + \lambda_j}\right)^2 
= \sum_{j=1}^{n-p} \left(1 - \frac{\tau}{\tau + \lambda_j}\right)^2
= n - 2\tau\text{tr}(\mathbf{V}) + \tau^2\text{tr}(\mathbf{V}^2),
\end{equation}
since the trace of $\bfV$ is $\sum (\tau + \lambda_j)^{-1} + p\tau^{-1}$ and the trace of  $\bfV^2$ is $\sum (\tau + \lambda_j)^{-2} + p\tau^{-2}$.

Now, square Equation (\ref{eqn:master-equation-for-V}) for $\mathbf{V}$,  take the trace, and multiply by $\tau^2$ to obtain
\begin{equation}\label{eqn:trace-V2}
\tau^2\text{tr}(\mathbf{V}^2) = \text{tr}(\bfB(\tau)^2) + 2\tau^{-1}\text{tr}(\{\mathbf{Q}_{01}(\tau)\}^{-1}\mathbf{Q}_{13}(\tau)) + \tau^{-2}\text{tr}\left( (\{\mathbf{Q}_{01}(\tau)\}^{-1}\mathbf{Q}_{12}(\tau))^2 \right).
\end{equation}

Now, in Equation~(\ref{eqn:T2}), substitute the expressions for $\tau\text{tr}(\mathbf{V})$ and $\tau^2\text{tr}(\mathbf{V}^2)$ from Equations~(\ref{eqn:trace-V}) and (\ref{eqn:trace-V2}) to obtain
\begin{eqnarray*}
T_2 & = & \text{tr}(\mathbf{I}_n - 2\bfB(\tau) + \bfB(\tau)^2) + \tau^{-2}\text{tr}((\{\mathbf{Q}_{01}(\tau)\}^{-1}\mathbf{Q}_{12}(\tau))^2) \\
& & - 2\tau^{-1}\text{tr}\left( \{\mathbf{Q}_{01}(\tau)\}^{-1}(\mathbf{Q}_{12}(\tau) - \mathbf{Q}_{13}(\tau)) \right).
\end{eqnarray*}

Finally, $\mathbf{I}_n - 2\bfB(\tau) + \bfB(\tau)^2 = (\mathbf{I}_n - \bfB(\tau))^2 = (\tau^{-1}\mathbf{D}^+\bfB(\tau))^2$. In addition, $\mathbf{Q}_{12}(\tau) - \mathbf{Q}_{13}(\tau) = \widetilde{\bfX}'\mathbf{D}^+\mathbf{M}^2(\mathbf{I}_n - \mathbf{M})\widetilde{\bfX} = \widetilde{\bfX}'\mathbf{D}^+\mathbf{M}^2(\tau^{-1}\mathbf{D}^+\mathbf{M})\widetilde{\bfX} = \tau^{-1}\mathbf{Q}_{23}(\tau)$. Therefore,
\begin{equation}\label{eqn:T2-final}
T_2 = \tau^{-2}\left\{\text{tr}((\mathbf{D}^+\bfB(\tau))^2) + \text{tr}((\{\mathbf{Q}_{01}(\tau)\}^{-1}\mathbf{Q}_{12}(\tau))^2) - 2\text{tr}\left( \{\mathbf{Q}_{01}(\tau)\}^{-1}\mathbf{Q}_{23}(\tau) \right)\right\}.
\end{equation}

Substituting in Equation~(\ref{eqn:referenceprior-tau-T1-T2}) the expressions for $T_1$ and $T_2$ from Equations~(\ref{eqn:T1-final}) and (\ref{eqn:T2-final}) proves the theorem.
\hfill $\square$

\end{Appendix}

\baselineskip=14pt \vskip 4mm\noindent

\bibliographystyle{chicago}
\bibliography{ref}

@article{Ferreira2021,
title = {Fast and scalable computations for {G}aussian hierarchical models with intrinsic conditional autoregressive spatial random effects},
journal = {Computational Statistics and Data Analysis},
volume = {162},
pages = {107264},
year = {2021},
issn = {0167-9473},
doi = {https://doi.org/10.1016/j.csda.2021.107264},
url = {https://www.sciencedirect.com/science/article/pii/S0167947321000980},
author = {Marco A. R. Ferreira and Erica M. Porter and Christopher T. Franck}
}

@article{lee2013,
  title={{CARBayes}: An {R} package for {B}ayesian spatial modeling with conditional autoregressive priors},
  author={Lee, Duncan},
  journal={Journal of Statistical Software},
  volume={55},
  number={13},
  pages={1--24},
  year={2013}
}

@inproceedings{best:etal:1999,
  title={Bayesian models for spatially correlated disease and exposure data},
  author={Best, N G and Arnold, R A and Thomas, Andrew and Waller, L A and Conlon, E M},
  booktitle={Bayesian Statistics 6: Proceedings of the Sixth Valencia International Meeting},
  volume={6},
  pages={131},
  year={1999},
  organization={Oxford University Press}
}

@article{wata:2010,
  title={Asymptotic equivalence of {B}ayes cross validation and widely applicable information criterion in singular learning theory},
  author={Watanabe, Sumio},
  journal={Journal of Machine Learning Research},
  volume={11},
  pages = {3571--3594},
  year={2010}
}

@article{cele:etal:2006,
  title={Deviance information criteria for missing data models},
  author={Celeux, Gilles and Forbes, Florence and Robert, Christian P and Titterington, D Mike},
  journal={Bayesian Analysis},
  volume={1},
  number={4},
  pages={651--673},
  year={2006},
  publisher={International Society for Bayesian Analysis}
}

@article{ferr:2019,
author={Ferreira, M. A. R.},
title={The Limiting Distribution of the {G}ibbs Sampler for the Intrinsic Conditional Autoregressive Model},
journal={Brazilian Journal of Probability and Statistics},
volume={33},
pages={734--744},
year={2019}
}

@article{besag1991,
  title={Bayesian image restoration, with two applications in spatial statistics},
  author={Besag, Julian and York, Jeremy and Molli{\'e}, Annie},
  journal={Annals of the Institute of Statistical Mathematics},
  volume={43},
  number={1},
  pages={1--20},
  year={1991},
  publisher={Springer}
}

@article{ferr:deol:2007,
  title={Bayesian reference analysis for {G}aussian {M}arkov random fields},
  author={Ferreira, Marco {A} {R} and De Oliveira, Victor},
  journal={Journal of Multivariate Analysis},
  volume={98},
  number={4},
  pages={789--812},
  year={2007},
  publisher={Elsevier}
}

@article{besag1974,
  title={Spatial interaction and the statistical analysis of lattice systems},
  author={Besag, Julian},
  journal={Journal of the Royal Statistical Society -- Series B},
  pages={192--236},
  year={1974},
  publisher={JSTOR}
}

@article{reich2006,
  title={Effects of residual smoothing on the posterior of the fixed effects in disease-mapping models},
  author={Reich, Brian J and Hodges, James S and Zadnik, Vesna},
  journal={Biometrics},
  volume={62},
  number={4},
  pages={1197--1206},
  year={2006},
  publisher={Wiley Online Library}
}

@article{Jin2007,
author = {Jin, Xiaoping and Banerjee, Sudipto and Carlin, Bradley P.},
title = {Order-free co-regionalized areal data models with application to multiple-disease mapping},
journal = {Journal of the Royal Statistical Society: Series B (Statistical Methodology)},
volume = {69},
number = {5},
pages = {817-838},
doi = {https://doi.org/10.1111/j.1467-9868.2007.00612.x},
url = {https://rss.onlinelibrary.wiley.com/doi/abs/10.1111/j.1467-9868.2007.00612.x},
eprint = {https://rss.onlinelibrary.wiley.com/doi/pdf/10.1111/j.1467-9868.2007.00612.x},
year = {2007}
}

@article{Lee2011,
title = "A comparison of conditional autoregressive models used in {B}ayesian disease mapping",
journal = "Spatial and Spatio-temporal Epidemiology",
volume = "2",
number = "2",
pages = "79 - 89",
year = "2011",
issn = "1877-5845",
doi = "https://doi.org/10.1016/j.sste.2011.03.001",
url = "http://www.sciencedirect.com/science/article/pii/S1877584511000049",
author = "Duncan Lee"
}

@article{Goicoa2018,
author = {Goicoa, T and Adin, A. and Ugarte, M. D. and Hodges, J. S.},
year = {2018},
pages = {749-770},
title = {In spatio-temporal disease mapping models, identifiability constraints affect {PQL} and {INLA} results},
volume = {32},
journal = {Stochastic Environmental Research and Risk Assessment},
doi = {https://doi.org/10.1007/s00477-017-1405-0}
}

@article{liu2016,
  title={Pre-surgical {fMRI} data analysis using a spatially adaptive conditionally autoregressive model},
  author={Liu, Zhuqing and Berrocal, Veronica J and Bartsch, Andreas J and Johnson, Timothy D and others},
  journal={Bayesian Analysis},
  volume={11},
  issue={2},
  pages={599--625},
  year={2016},
  publisher={International Society for Bayesian Analysis}
}

@article{mercer2015,
  title={Space--time smoothing of complex survey data: Small area estimation for child mortality},
  author={Mercer, Laina D and Wakefield, Jon and Pantazis, Athena and Lutambi, Angelina M and Masanja, Honorati and Clark, Samuel and others},
  journal={The Annals of Applied Statistics},
  volume={9},
  number={4},
  pages={1889--1905},
  year={2015},
  publisher={Institute of Mathematical Statistics}
}

@incollection{lyu2024spatial,
  title={Spatial Modeling of Imaging and Electrophysiological Data},
  author={Lyu, Rongke and Guindani, Michele and Vannucci, Marina},
  booktitle={Statistical Methods in Epilepsy},
  editor={Chiang, Sharon and Rao, Vikram R. and Vannucci, Marina},
  chapter={10},
  year={2024},
  publisher={Chapman and Hall/CRC},
  address={Boca Raton, FL},
  doi={10.1201/9781003254515-10}
}

@book{kolaczyk2020statistical,
  title={Statistical Analysis of Network Data with R},
  author={Kolaczyk, Eric D. and Cs{\'a}rdi, G{\'a}bor},
  edition={2nd},
  year={2020},
  publisher={Springer}
}

@article{VerHoef2018,
	doi = {10.1002/ecm.1283},
	url = {https://doi.org/10.1002/ecm.1283},
	year = 2018,
	publisher = {Wiley},
	volume = {88},
	number = {1},
	pages = {36--59},
	author = {Jay M. {Ver Hoef} and Erin E. Peterson and Mevin B. Hooten and Ephraim M. Hanks and Marie-Jos{\`{e}}e Fortin},
	title = {Spatial autoregressive models for statistical inference from ecological data},
	journal = {Ecological Monographs}
}

@book{magn:neud:1999,
	Address = {Chichester},
	Author = {J. R. Magnus and H. Neudecker},
	Edition = {Revised},
	Publisher = {Wiley},
	Title = {Matrix Differential Calculus with Applications in Statistics and Econometrics},
	Year = {1999}}

@Manual{port:etal:2025,
    title = {ref.ICAR: Objective Bayes Intrinsic Conditional Autoregressive Model for
Areal Data},
    author = {Erica M. Porter and Matthew J. Keefe and Christopher T. Franck and Marco A.R. Ferreira},
    year = {2025},
    note = {R package version 2.0.2},
    url = {https://CRAN.R-project.org/package=ref.ICAR},
    doi = {10.32614/CRAN.package.ref.ICAR},
  }

@article{Porter2024,
  author    = {Porter, Erica M. and Franck, Christopher T. and Ferreira, Marco A. R.},
  title     = {Objective {Bayesian} Model Selection for Spatial Hierarchical Models with Intrinsic Conditional Autoregressive Priors},
  journal   = {Bayesian Analysis},
  year      = {2024},
  volume    = {19},
  number    = {4},
  pages     = {985--1011},
  doi       = {10.1214/23-BA1375}
}

@article{scott2010,
author = "Scott, James G. and Berger, James O.",
doi = "10.1214/10-AOS792",
fjournal = "The Annals of Statistics",
journal = "Annals of Statistics",
number = "5",
pages = "2587--2619",
publisher = "The Institute of Mathematical Statistics",
title = "Bayes and empirical-{B}ayes multiplicity adjustment in the variable-selection problem",
url = "https://doi.org/10.1214/10-AOS792",
volume = "38",
year = "2010"
}

@article{deoliveira2007,
  title={Objective {B}ayesian analysis of spatial data with measurement error},
  author={De Oliveira, Victor},
  journal={Canadian Journal of Statistics},
  volume={35},
  number={2},
  pages={283--301},
  year={2007},
  publisher={Wiley Online Library}
}

@Manual{r:2025,
    title = {R: A Language and Environment for Statistical Computing},
    author = {{R Core Team}},
    organization = {R Foundation for Statistical Computing},
    address = {Vienna, Austria},
    year = {2025},
    url = {https://www.R-project.org/},
  }

@article{deol:2011,
  title={Maximum likelihood and restricted maximum likelihood estimation for a class of {G}aussian {M}arkov random fields},
  author={De Oliveira, Victor and Ferreira,  Marco {A} {R}},
  journal={Metrika},
  volume={74},
  number={2},
  pages={167--183},
  year={2011},
  publisher={Springer}
}

@article{Liang2008,
author = {Feng Liang and Rui Paulo and German Molina and Merlise A Clyde and Jim O Berger},
title = {Mixtures of g Priors for {Bayesian} Variable Selection},
journal = {Journal of the American Statistical Association},
volume = {103},
number = {481},
pages = {410-423},
year  = {2008},
publisher = {Taylor & Francis},
doi = {10.1198/016214507000001337},
URL = {https://doi.org/10.1198/016214507000001337},
eprint = {https://doi.org/10.1198/016214507000001337}
}

@Article{berg:oliv:sans:2001,
  author = 	 {James O. Berger and Victor de Oliveira and Bruno Sans\'o},
  title = 	 {Objective {B}ayesian Analysis of Spatially Correlated Data},
  journal = 	 {Journal of the American Statistical Association},
  year = 	 {2001},
  volume = 	 {96},
  number = 	 {456},
  pages = 	 {1361--1374}
}

@article{Barbieri2004,
  author  = {Barbieri, Maria Maddalena and Berger, James O.},
  title   = {Optimal predictive model selection},
  journal = {Annals of Statistics},
  Volume =  {32}, 
  Pages = {870--897},
  year    = {2004}
}

@article{keef:ferr:fran:2019,
  title={Objective {B}ayesian analysis for {G}aussian hierarchical models with intrinsic conditional autoregressive priors},
  author={Keefe, Matthew J and Ferreira, Marco {A} {R} and Franck, Christopher T},
  journal={Bayesian Analysis},
  volume={14},
  pages={181--209},
  year={2019}
}

@article{keef:ferr:fran:2018,
  title={On the formal specification of sum-zero constrained intrinsic conditional autoregressive models},
  author={Keefe, Matthew J and Ferreira, Marco {A} {R} and Franck, Christopher T},
  journal={Spatial Statistics},
  volume={24},
  pages={54--65},
  year={2018},
}

@article{ohagan1995,
  title={Fractional {B}ayes factors for model comparison},
  author={O'Hagan, Anthony},
  journal={Journal of the Royal Statistical Society. Series B (Methodological)},
  pages={99--138},
  year={1995},
  publisher={JSTOR}
}

@article{spiegelhalter2002,
  title={Bayesian measures of model complexity and fit},
  author={Spiegelhalter, David J and Best, Nicola G and Carlin, Bradley P and Van Der Linde, Angelika},
  journal={Journal of the Royal Statistical Society: Series B (Statistical Methodology)},
  volume={64},
  number={4},
  pages={583--639},
  year={2002},
  publisher={Wiley Online Library}
}

\end{document}